\documentclass[journal]{IEEEtran}
\usepackage{lipsum}
\usepackage{setspace}

  \usepackage[pdftex]{graphicx}
  \usepackage{caption}
  \usepackage{subcaption}
  \usepackage{tablefootnote}

\usepackage{tikz, pgfplots}

\usepackage{amsmath}
\usepackage{amssymb}
\usepackage{amsthm} 
\usepackage{algorithm}
\usepackage{algorithmic}

\usepackage{array}
\usepackage{dblfloatfix}

\usepackage{xcolor}
\usepackage{mathabx}
\usepackage{enumitem}
\usepackage[symbol, flushmargin]{footmisc}

\usepackage{url}

\usepackage{multirow}
\usepackage{booktabs}

\usepackage{nomencl}
\usepackage{etoolbox}
\makenomenclature
\renewcommand\nomgroup[1]{%
  \ifstrequal{#1}{P}{\vspace{10pt}\item[\textbf{Parameters}]}{%
  \ifstrequal{#1}{V}{\vspace{10pt}\item[\textbf{Variables}]}{}}{%
  \ifstrequal{#1}{S}{\vspace{10pt}\item[\textbf{Sets}]}{}}%
}
\usepackage[normalem]{ulem}
\definecolor{addblue}{rgb}{0.1,0,0.8}

\definecolor{darkgrn}{rgb}{0, 0.75, 0}

\usepackage[none]{hyphenat}
\usepackage{tabularx}

\usepackage[belowskip=1pt,aboveskip=1pt]{caption}
\usepackage{titlesec}
\titlespacing\section{0pt}{2pt plus 2pt minus 2pt}{2pt plus 2pt minus 2pt}
\titlespacing\subsection{0pt}{2pt plus 2pt minus 2pt}{2pt plus 2pt minus 2pt}
\titlespacing\subsubsection{0pt}{2pt plus 2pt minus 2pt}{2pt plus 2pt minus 2pt}

\usepackage{nccmath}
\usepackage{xpatch}
\xpatchcmd{\NCC@ignorepar}{%
\abovedisplayskip\abovedisplayshortskip}
{%
\abovedisplayskip0.2\abovedisplayshortskip%
\belowdisplayskip0.2\belowdisplayshortskip}
{}{}

\begin{document}

\thinmuskip=0mu
%

\title{Scalable Two-Stage Stochastic Optimal Power Flow via Separable Approximation}

%
%
%

\author{Shishir~Lamichhane,~\IEEEmembership{Student Member,~IEEE,}
        Abodh~Poudyal,~\IEEEmembership{Member,~IEEE,} Nicholas~R.~Jones, Bala~Krishnamoorthy,
        and~Anamika~Dubey,~\IEEEmembership{Senior~Member,~IEEE}
\thanks{
          This work is supported in part by the National Science Foundation (NSF) under Career Award Number 1944142, and in part by the U.S. Department of Energy’s Office of Energy Efficiency and Renewable Energy (EERE) under the Solar Energy  Technologies Office Award Number DE-EE00010424. }
\thanks{Shishir Lamichhane, Abodh Poudyal and Anamika Dubey  are with the School of Electrical Engineering and Computer Science, Washington State University, Pullman, WA;
  Nicholas R.Jones and Bala Krishnamoorthy are with the Department of Mathematics and Statistics, Washington State University, Vancouver, WA;
  (corresponding author email: shishir.lamichhane@wsu.edu).}
  \vspace{-20pt}
}

\maketitle

\begin{abstract}

  This paper proposes a Separable Projective Approximation Routine -- Optimal Power Flow (SPAR-OPF) framework for solving two-stage stochastic optimization problems in 
  power systems.
  The framework utilizes a separable piecewise linear approximation of the value function and learns the function based on sample sub-gradient information. 
  We present two formulations to model the learned value function, and compare their effectiveness.
  Additionally, an efficient 
  statistical method is introduced to assess the quality of the obtained solutions. The effectiveness of the proposed framework is validated using distributed generation siting and sizing problem in three-phase unbalanced power distribution systems as an example. 
  Results show that the framework approximates the value function with over 98\% accuracy and provides high-quality solutions with an optimality gap of less than 1\%. The framework scales efficiently with system size, generating high-quality solutions in a short time when applied to a 9500-node distribution system with 1200 scenarios, while the extensive formulations and progressive hedging failed to solve the problem.
\end{abstract}

\begin{IEEEkeywords}
Decision-making under uncertainty, power systems planning, stochastic optimal power flow, value function approximation.
\end{IEEEkeywords}

\section{Introduction}

\IEEEPARstart{T}{}he reliable, secure, and economical operation of power systems relies on solving a variety of optimization problems. Relevant applications include economic dispatch (ED), unit commitment (UC), market clearing, reserve planning, system restoration, expansion planning, and the optimal placement of distributed generation~\cite{OPF_survey_latest}. These optimization challenges encompass both transmission and distribution systems, with their complexity depending on factors such as problem type, decision time horizon, and the nature of the decision variables. Traditionally, grid optimization focused on a limited set of credible scenarios, where the nonlinearity of power flow and integer variables mainly posed significant challenges. However, the integration of variable energy resources, flexible/controllable loads, and pronounced weather-related uncertainties have increased the complexity of grid operations significantly~\cite{more_scenarios_to_include_latest}. Modern grids must now be optimized under a much broader range of scenarios, making optimization under uncertainty a far greater challenge that demands scalable and robust methods.

   Optimization under uncertainty has been widely studied in the power systems domain, particularly in relation to various Optimal Power Flow (OPF)-based problems. In this context, the power systems community has explored a range of formulations, including stochastic, robust, distributionally robust, and chance-constrained optimization, as well as their hybrid forms~\cite{StochasticReview_LP}. Many real-world power system challenges, such as UC, ED, proactive dispatch, and investment planning, naturally lend themselves to a two-stage decision-making framework~\cite{uncertainty_modeling_method_review}. In this setup, the first stage involves making non-anticipative decisions, while the second stage optimizes operational outcomes once uncertainty is realized. However, for large-scale systems, solving these two-stage problems becomes computationally intractable, particularly in resource-constrained environments~\cite{LineRoaldReviewStochastic}. This presents a critical gap in current methods, as the size and complexity of real-world power systems continue to grow. \emph{In this paper, we present a novel, computationally efficient, and scalable framework for solving two-stage stochastic optimization problems in power systems.} The proposed method leverages a separable piecewise linear approximation of value function to achieve near-optimal solutions, significantly enhancing tractability for large-scale systems while maintaining high solution quality.
   
\subsection{Background and Related Work}

Traditional approaches for solving two-stage optimization problems include chance-constrained optimization, robust optimization, and stochastic optimization. Chance-constrained optimization requires knowledge of uncertain parameter distributions and generally results in an optimization problem of higher complexity (e.g. non-convex problems, even when the original problem is convex)~\cite{LineRoaldReviewStochastic}. 
Robust optimization tends to be overly conservative, resulting in inefficient resource utilization. In contrast, stochastic optimization directly uses scenario sets as input, considering recourse actions in the second stage to ensure feasibility, thus eliminating the need to model or approximate uncertainty distributions. This motivates us to consider an approach based on stochastic optimization.

 Ideally, two-stage stochastic optimization involves solving a deterministic extensive formulation (EF) where multiple copies of the second stage problem are added to the model for each scenario~\cite{resilience_driven_deterministic}.
 Solving such an extensive form is computationally intractable, as the size of the decision space increases exponentially with the number of scenarios.
  With the increasing level of uncertainty, it is not advisable to excessively reduce the number of scenarios just for computational purposes, as a reasonable number of scenarios need to be considered in order to capture the uncertainty; otherwise, the obtained solutions may be sub-optimal or infeasible~\cite{more_scenarios_to_include_latest}. 
Traditional approaches to tackle scalability challenges include (1) solving the EF multiple times with reduced scenario size, e.g., using a Sample Average Approximation (SAA), (2) using a stage-wise decomposition, e.g., Benders decomposition (BD), and (3) using a scenario-wise decomposition, e.g., progressive hedging (PH)~\cite{LineRoaldReviewStochastic}.
SAA involves solving the EF problem multiple times while taking a smaller sample size from the set of scenarios to approximate the value function.
 For example, SAA was used to solve the unit commitment problem considering the 118-bus system with a sample size of 50 scenarios~\cite{SAA_latest_118_bus_50_scenarios_latest}.
However, for a realistic system size, the problem poses significant computational challenges even with a relatively small number of scenarios. 
BD and its variants have been utilized in solving various OPF problems. For instance, BD was used to solve the unit commitment problem for the IEEE 118 bus system considering 10 scenarios~\cite{stochastic_unit_commitment_10_scenarios_BD}.
Stochastic decomposition, a variant of BD, has been applied in~\cite{stochastic_decomposition_economic_dispatch_sub_hourly_stochastic_decomposition_10_scenarios} for solving the economic dispatch problem considering a maximum system size of 96 buses and a maximum number of 100 scenarios. As the number of scenarios grows, the demand for generating valid cuts increases rapidly, making the problem intractable and unsuitable for larger systems with numerous scenarios.
It was also observed that BD often leads to long computational times, excessive use of memory resources, poor feasibility and optimality cuts, and slow convergence (tail-off effect)~\cite{software_packages_review}.
PH has been applied to distributed generation (DG) siting and sizing on a 123-bus system with 400 scenarios~\cite{2020_resilience_oriented_DG_siting_sizing_PA}, and to proactive dispatch and investment planning on a 2000-bus system with 512 scenarios~\cite{proactive_dispatch_PA}.
However, in real-world applications, the default settings for PH often pose challenges~\cite{bynum2021pyomo}. 
Moreover, due to its decompositional nature, PH is most effective in parallel architectures with multi-core devices, thereby restricting its broader applicability.
 Separate from these methods, the optimization community has explored the use of learning-based approaches to solve two-stage stochastic optimization problems. One popular approach includes a separable approximation of a value function where the coordinate functions are approximated using a learning-based technique. A static piecewise linear approximation was proposed in~\cite{1990PowellStaticVFA}, but this method, being not adaptive, limits the quality of the solution. Later, a  Separable Projective Approximation Routine (SPAR) was introduced, which adaptively constructs a sequence of piecewise linear functions using sample sub-gradient information~\cite{VFA_powell}. It proved to be optimal when the value function is separable while also providing a near-optimal solution for the non-separable case, which proved faster than BD~\cite{VFA_powell}. Unlike popular black-box neural network models, these methods maintain the explainability of the learned functions and hence are preferable for power systems, where interpretability of decision-making processes are crucial for safety, reliability, and regulatory compliance.  

\subsection{Contributions}
We introduce \emph{Separable Projective Approximation Routine - Optimal Power Flow (SPAR-OPF)}, an efficient framework designed to solve two-stage stochastic optimization problems involving OPF in power systems. Leveraging a SPAR-based approach, our framework scales effectively with system size and can be implemented in resource-constrained environments. To the best of our knowledge, this is the first framework to solve OPF-based two-stage stochastic optimization problems using separable projective approximation of value function.
From the application perspective, we consider the following two-stage setting: Stage-1 makes \emph{ex-ante} decisions based on forecasts, and Stage-2 makes \emph{ex-post} decisions after the realization of the event/uncertainty. Stage-1 can be thought of as a long-term or operational planning problem, and Stage-2 represents real-time operations. The general assumption is that planning decisions should be optimized considering optimal real-time operations. 
The computational performance of SPAR-OPF indicates its applicability for solving OPF under uncertainty for realistic power systems.
The following are the main contributions of this paper. 


\begin{itemize}[leftmargin=*]
\item \emph{SPAR-OPF framework for two-stage stochastic optimization:} We propose a framework to solve two-stage stochastic OPF problems in power systems based on separable projective approximation of value function. We present two formulations: the lambda method and the epigraph method to model the learned value function and compare their performance.
\item \emph{Posterior statistical bounds for optimal solutions:} We propose a computationally efficient statistical method that uses a posterior analysis to obtain lower and upper bounds for the optimal objective function value as measures of quality of the obtained solutions. 
\item \emph{Validation and scalability assessment:}
We validate and assess the scalability of SPAR-OPF on a two-stage OPF problem optimizing the siting and sizing decisions of PV-based distributed generation (DG) for enhanced and efficient grid operation. 
The proposed SPAR-OPF approach is compared with approaches using Extensive Formulations (EF) and the Progressive Hedging (PH) algorithm. Finally, the scalability of SPAR-OPF is demonstrated on a 9500-node distribution test system.  


%

\end{itemize}

\vspace{0.3cm}
\section{Separable Projective Approximation
Routine -- Optimal Power Flow (SPAR-OPF) Framework}\label{sec:Modeling}
\subsection{Problem Formulation}

We present SPAR-OPF as a general two-stage stochastic optimization model with optimal planning in the first stage and optimal operation in the second stage. 
Subsequently, specific details for the framework using DG siting and sizing as an example of a two-stage stochastic OPF problem are presented. 
The first stage problem is denoted as $\textbf{P}_{\textbf{g1}}$ and is defined in~(\ref{eq:general opf first stage}).
It is important to mention that we do not impose any restrictions on the nature or structure of $\textbf{P}_{\textbf{g1}}$. The objective function consists of a first-stage cost $H(x,y)$ that is dependent only on the first-stage decision variables $x$ and $y$, and the expected recourse cost $Q(x, \xi)$.
We obtain $Q(x, \xi)$ by solving the second stage OPF problem denoted by $\textbf{P}_{\textbf{g2}}$ with $z^\xi$ as decision variables, as given in~(\ref{eq:general opf second stage}). OPF has been shown to be exact under various convex relaxations and approximations~\cite{OPF_relaxations}. We assume such convex problems in the second stage with particular structures (see Remark~\ref{pre:remark problem scope}) that the scientific community has widely considered when solving power system optimization problems under uncertainty~\cite{StochasticReview_LP, resilience_driven_deterministic, stochastic_unit_commitment_10_scenarios_BD, resilience_oriented_DG_siting_sizing}. Here, $x$ is the set of linking variables involved in the second stage for each scenario $\xi$. Note that the constraints in the second stage problem are separated into two types: one dependent on the first stage decision variable $x$, denoted as linking constraints~(\ref{eq:general second stage constraint I}), and the rest of the constraints (\ref{eq:general second stage constraint II-a})--(\ref{eq:general second stage constraint II-b}) that do not directly involve the linking variables. Most general two-stage power grid optimization problems can be reformulated to fit the problem structure defined in~(\ref{eq:general opf first stage})--(\ref{eq:general second stage constraint II-b}).





\useshortskip
\begin{equation}
        (\textbf{P}_{\textbf{g1}}):  \quad \min_{x,y} ~ H(x,y) + \mathbb{E}_{\xi}[Q(x, \xi)] \label{eq:general opf first stage}
\end{equation}
\useshortskip
\useshortskip
 \begin{equation}
    \text{\quad s.t. } \quad   \text{\{planning constraints\}} \label{eq:general first stage constraint}
\end{equation}






\vspace{-1.2em}
 \smallskip
\begin{equation}
 (\textbf{P}_{\textbf{g2}}):\quad  Q(x,\xi) = ~ \min_{z^\xi}~ f_0(z^\xi) \label{eq:general opf second stage} \hspace{0.4in}
\end{equation}

\useshortskip


\useshortskip
\begin{equation}
         \text{\quad s.t. }\quad  A^\xi z^\xi \leq x\   \text{\{linking constraints\}} \label{eq:general second stage constraint I}
\end{equation}

\useshortskip
 \begin{equation}
  f_i(z^\xi) \leq 0\quad \forall\ i \in 1,2,...m \label{eq:general second stage constraint II-a}
\end{equation}

\useshortskip
 \begin{equation}
  B^\xi z^\xi = e^\xi\quad \label{eq:general second stage constraint II-b}
\end{equation}

\AtBeginEnvironment{proof}{\vspace{-8pt}}
\AtEndEnvironment{proof}{\vspace{-4pt}}

\newtheorem{proposition}{Proposition}
\newtheorem{corollary}{Corollary}
\newtheorem{lemma}{Lemma}

\begin{proposition} \label{pre:preposition 1}
If the second stage problem is convex with continuous recourse variables, $Q(x,\xi)$ is a convex function of the linking variable $x$.
\end{proposition}

\begin{proof}
Let us consider $\alpha\geq 0$, $\beta \geq 0$ and $\nu$ be the dual variables corresponding to (\ref{eq:general second stage constraint I}), (\ref{eq:general second stage constraint II-a}) and (\ref{eq:general second stage constraint II-b}) respectively. Note that $x$ is treated as parameter while solving second stage problem separately. Then, the Lagrangian of  $\textbf{P}_{\textbf{g2}}$ is given as 
\vspace{3pt}
\useshortskip
\begin{equation}
\begin{gathered}
    L(z^\xi, \alpha,\beta,\nu; x,\xi) = f_0(z^\xi) + \alpha^T(A^\xi z^\xi - x) + \\\sum_{i = 1}^{m} \beta_if_i(z^\xi) + \nu^T (B^\xi z^\xi - e^\xi) \label{eq:proof lagrange expression}
    \end{gathered}
\end{equation}
\noindent The Lagrangian dual function $q(\alpha,\beta,\nu; x, \xi)$ and the associated dual problem optimal value $q^D(x,\xi)$ are then obtained as 

\useshortskip
\begin{equation}
  \begin{gathered}
    q(\alpha,\beta,\nu; x, \xi) = \inf_{z^\xi}~  f_0(z^\xi) + \alpha^T(A^\xi z^\xi - x) + \quad\quad \\\sum_{i = 1}^{m} \beta_if_i(z^\xi) + \nu^T (B^\xi z^\xi - e^\xi) \quad \text{\rm and} \label{eq:proof lagrange dual expression}
  \end{gathered}
\end{equation}

\useshortskip
\begin{equation}
  \begin{gathered}
    q^D(x,\xi) = \max_{\alpha, \beta,\nu}~ q(\alpha,\beta,\nu; x, \xi)\label{eq:lagrange_dual_problem}
  \end{gathered}
\end{equation}

The Lagrangian dual problem (\ref{eq:lagrange_dual_problem}) is convex.
With the assumption that Slater's condition is satisfied for the inequalities in (\ref{eq:general second stage constraint I}) and (\ref{eq:general second stage constraint II-a}), which is a weak assumption in power systems problems~\cite{OPF_relaxations}, 
we get by strong duality that $Q(x,\xi) = q^D(x,\xi)$. Hence, $Q(x,\xi)$ is a convex function of $x$.
\end{proof}

\begin{corollary}
\label{pre:corollary 1}
If the second stage problem is a linear problem (LP), $Q(x,\xi)$ is a piecewise linear convex function of the linking variable $x$.
\end{corollary}

\begin{proof}
    Immediately following the proof of Proposition 1, we see that $q(\alpha,\beta,\nu; x, \xi)$ is contained in a polyhedra in $(\alpha, \beta, \nu)$ due to second stage problem being LP. Therefore, $q^D(x,\xi)$ is the maximum of finitely many linear functions in $x$ and thus a piecewise linear convex function of $x$. Hence, $Q(x,\xi)$ is a piecewise linear convex function of $x$.
\end{proof}





\begin{lemma}
    \label{pre:lemma 1}
    If the linking variable $x$ is integral,  then there exists a piecewise linear function $F(x,\xi)$ such that $F(x,\xi)$ is equal to $Q(x,\xi)$ for all integer $x$. Consequently, the optimal solutions of $Q(x,\xi)$ over integer values of $x$ coincide with those of $F(x,\xi)$.

\end{lemma}

\begin{proof}
    Consider the function below.
    \useshortskip
    \begin{equation*}
        F(x,\xi)=Q(\lfloor x\rfloor,\xi)+\sum_{i=1}^n \left[Q(\lfloor x\rfloor + e_i, \xi) - Q(\lfloor x\rfloor, \xi)\right]\left(x_i - \lfloor x_i\rfloor\right)
    \end{equation*}

    where $\lfloor x \rfloor$ denotes the component-wise floor of $x$, $x_i$ is $i$-th component of $x$ and $e_i$ 
is the $i$-th unit vector.

    \noindent Observe that $F$ is a piecewise linear extension of $Q$ with integer breakpoints. Trivially, $F(x,\xi)=Q(x,\xi)$ for all integer $x$ as desired. Furthermore, the optima of a piecewise linear function with integer break points occur at integer values, implying that the optima of the two functions coincide.
\end{proof}




\newtheorem{remark}{Remark}

\begin{remark}
\label{pre:remark problem scope}
 We assume the function $Q(x,\xi)$ to be a piecewise linear convex function in the linking variable $x$. Corollary~\ref{pre:corollary 1} ensures that this property holds for LP-based second-stage problems, which are widely used in two-stage stochastic power system problems~\cite{StochasticReview_LP}. Lemma~\ref{pre:lemma 1} further extends the framework to integer linking variables with a general convex second-stage problem, including (but not limited to) system expansion planning~\cite{resilience_driven_deterministic}, unit commitment~\cite{stochastic_unit_commitment_10_scenarios_BD}, and DG siting and sizing~\cite{resilience_oriented_DG_siting_sizing}. For other convex formulations, Proposition~\ref{pre:preposition 1} guarantees convexity in $x$ but not necessarily piecewise linearity; in such a case, $Q(x,\xi)$ may be approximated by a piecewise linear function so that the framework remains applicable.
\end{remark}

\begin{remark}
 Under the common assumption of exogenous uncertainty~\cite{LineRoaldReviewStochastic}, and that $Q(x,\xi)$ is piecewise linear convex in $x$, $\mathbb{E_\xi} [Q(x,\xi)]$ is a piecewise linear convex function as it is the linear combination, with probabilities of scenarios as coefficients, of piecewise linear convex functions.
\end{remark}



    


\subsection{Value Function Approximation (VFA):}

The value function $\mathbb{E_\xi} [Q(x,\xi)]$ in objective function (\ref{eq:general opf first stage}) of first stage problem $\textbf{P}_{\textbf{g1}}$  is a function of linking variables $x$ and scenarios $\xi$ realized in the second stage. 
 Two-stage problems, by nature, are not separable. However, separable approximations of  value function have been widely
used and found to provide near-optimal solution, typically by assuming separability with respect to linking variables when they appear affinely in the second-stage constraints ~\cite{1990PowellStaticVFA,VFA_powell}. We also consider such an approximation, and the value function is written as a summation of coordinate functions $g_i(x_i)$
as given in~(\ref{eq:SPAR coordinate sum}) where $\mathcal{I}$ is the set of coordinate functions. 

\useshortskip
\begin{equation}\label{eq:SPAR coordinate sum}
    G(x) = \mathbb{E}_{\xi}[Q(x,\xi)] = \sum_{i = 1}^{|\mathcal{I}|}g_i(x_i)
\end{equation}

Let us assume $g_i(x_i)$ is a piecewise linear convex function with integer breakpoints in $x_i$ dimension.
Let $m_i^l$ be its slope between points $l\ \in \mathcal{L}_i$ and $l+1$ as given in~(\ref{eq: slope SPAR}). The value of the function at any point $x_i$ can be obtained using (\ref{eq:learned function value SPAR}) where $u$ is an integer point.
\useshortskip
\begin{equation}
    m_i^l  = g_i(l+1)\ -\ g_i(l) \label{eq: slope SPAR}
\end{equation}


\useshortskip
\begin{equation}
  \hspace*{-0.12in}
\begin{gathered}
g_i(x_i) =
\begin{cases}
g_i(0) - \sum_{l = u}^{-1} m_i^l + m_{i}^u (x_i - u) \, \text{ for } x_i < 0  \\
g_i(0) + \sum_{l = 0}^{u-1} m_i^l + m_{i}^u (x_i - u) \, \text{ for } x_i \geq 0  
\end{cases} \\
\text{where}\ u \leq x_i < u+1.
\label{eq:learned function value SPAR}
\end{gathered}
\hspace*{-0.12in}
\end{equation}





It is to be noted that, evaluating this function requires its value at origin, $g_i(0)$, which may not be available \emph{a priori}. The decision points, however, are not affected by this constant term, and hence, $g_i(0)=0$ can be assumed. 

Once the coordinate functions are learned, a formulation for (\ref{eq:SPAR coordinate sum}) is constructed, and the first stage problem  $\textbf{P}_{\textbf{g1}}$ can then be solved by considering only the first stage feasible set. In this regard, we propose two formulations. First, we introduce binary variables to approximate the function using a well-known polyhedral approach to model piecewise linear functions~\cite{lambda_method}, which we refer to as the \emph{lambda method} $(\textbf{L}_\textbf{1})$. However, introducing a large number of binary variables can significantly increase computational complexity. To mitigate this challenge, we propose a second approach that we refer to as the \emph{epigraph method} $(\textbf{E}_\textbf{1})$, which leverages the convexity of the value function by incorporating a continuous epigraph variable for each coordinate function.

\subsubsection{Lambda Method}\label{secsec: lambda method}


Once the slope of the coordinate function $g_i$ is learned, the function can be computed using (\ref{eq:learned function value SPAR}). We introduce binary variables $\lambda_i^l$ associated with each decision point $l \in \mathcal{L}$ of the coordinate function $g_i$ to approximate~(\ref{eq:SPAR coordinate sum}) as specified in the objective function in~(\ref{eq:lambda method obj}). Here, $g_i^l$ represents the value of the function $g_i$ at point $l$. 
Eq.~(\ref{eq:lambda method const1}) ensures that each coordinate function has a single decision point. The corresponding linking decision variable $x_i$ is then given by~(\ref{eq:lambda method const2}).





\useshortskip
\begin{equation}
       (\textbf{L}_\textbf{1}):\quad  \min~ H(x,y) +  \sum_{l \in \mathcal{L}_i}  \sum_{i \in \mathcal{I}} g_i^l \lambda_i^l \label{eq:lambda method obj}
       \end{equation}


\useshortskip
\begin{equation}
         \text{\hspace*{0.6in} s.t. } ~~  \sum_{l \in\mathcal{L}\,} \lambda_i^l = 1 ~~ \forall\ i \in \mathcal{I} \quad\quad \label{eq:lambda method const1}
\end{equation}

\useshortskip
\begin{equation}
          \sum_{l \in \mathcal{L}} \lambda_i^l l = x_i ~~\forall\ i \in \mathcal{I} \label{eq:lambda method const2}
\end{equation}

\useshortskip 
\begin{equation*}
  \text{planning constraints (\ref{eq:general first stage constraint})}
\end{equation*}

\useshortskip 
\begin{equation*}
     \lambda_i^l \in \{0,1\}, ~ \forall\ l \in \mathcal{L}_i, \forall\ i \in \mathcal{I}
\end{equation*}


\begin{remark} 
  Constraint (\ref{eq:lambda method const2}) ensures that $x_i$ is integral in any feasible solution even when it is modeled as a continuous variable, since $l \in \mathcal{L}_i$ is an integer parameter and $\lambda_i^l \in \{0,1\}$.
\end{remark} 

\subsubsection{Epigraph method}\label{secsec:epigraph method}


Since the learned function is piecewise linear and convex with integer breakpoints, minimizing the function~(\ref{eq:general opf first stage}) after replacing the expectation term with~(\ref{eq:SPAR coordinate sum}) is equivalent to minimizing the function~(\ref{eq:epigraph obj}) with additional constraints~(\ref{eq:epigraph const 1}). 
The function $g_i$ is approximated by the affine function using slope information $m_i^l$ at each integer breakpoint $l$. This function is then upper bounded by the epigraph variable $t_i$. 



\useshortskip
\begin{equation}
       (\textbf{E}_\textbf{1}):\quad  \min\ H(x,y) +  \sum_{i\in \mathcal{I}}\, t_i \label{eq:epigraph obj}
\end{equation}


\useshortskip
\begin{equation}
    \text{s.t. }  \quad     m_i^l\,(x_i -  l) + g_i^l \leq t_i ~ \forall\ l \in \mathcal{L}_i, \forall\ i \in \mathcal{I} \label{eq:epigraph const 1}
\end{equation}

\useshortskip
\begin{equation*}
          \text{planning constraints (\ref{eq:general first stage constraint})}
\end{equation*}


\vspace{0.2em}
\subsection{Separable Projective Approximation Routine (SPAR)}\label{sec:SPAR}

In this section, we present the SPAR algorithm from~\cite{VFA_powell} for learning the coordinate functions $g_i$. Any term with the superscript $(k)$ refers to its value at iteration $k$. For simplicity, we represent $g_i(x_i)$ as $g(x)$ and $m_i^l$ as $m_l$, which is the slope of function $g(x)$ between points $l$ and $l+1$. 
The overall goal is to learn and update the slopes iteratively such that the sequence $\{m^{(k)}\}$ converges to the true slope $\overline{m}$. In other words, the function $\{g^{(k)}\}$  constructed from the learned slopes using (\ref{eq:learned function value SPAR}) is the functional approximation of the true function $g$. Algorithm~\ref{algo:SPAR general} specifies how to learn the slopes of the coordinate function $g$. First, the slope vector $m^{(1)}$ is initialized to the zero vector and the maximum number of iterations is chosen as $k^{\max}$. Then, the first stage decision point $l^{(k)}$ is obtained by solving the first stage problem $\textbf{P}_{\textbf{g1}}$ after replacing the expected function by approximated function $G^{(k)}$ using~(\ref{eq:SPAR coordinate sum}). Notice that in initial iterations, exploration can be employed by randomly choosing decision points, rather than exploiting the optimal decisions in every iteration. After this, a random variable $\gamma^{(k)}$ is observed, which is $H_k$ measurable,
and its expectation is equal to the true slope $\overline{m}_{l^{(k)}}$. As discussed in~\cite{VFA_powell}, $\gamma^{(k)}$ satisfies this condition when Lagrange multipliers are taken after solving the second-stage problem for the given first-stage decision $l^{(k)}$. Next, the slope is updated at point $l^{(k)}$ as shown in step 5. Step size $\{\alpha^{(k)}\}$ is also $H_k$ measurable and assumed to satisfy \(\sum_{k = 1}^\infty \alpha^{(k)} = \infty\) and \(\sum_{k = 1}^\infty (\alpha^{(k)})^2 < \infty\). The updated slope vector $n^{(k)}$ is then projected on a convex set $\mathcal{M}$ given by (\ref{eq: proj const}) to give slope vector $m^{(k+1)}$ for the next iteration to ensure the convexity of the learned function. Finally, the function $G^{(k)}$ is updated to give $G^{(k+1)}$ and the process continues until $k^{\max}$ iterations are reached or the termination criteria are met. One of the termination criteria considered in this paper is that the algorithm terminates when the moving average of the objective value at the current iteration falls within an $\epsilon$-radius of the moving average of the previous iteration.

The work in~\cite{VFA_powell} had assumed relatively complete recourse, which means the second stage is always feasible for any first-stage decisions and slope can be updated with Lagrange multipliers. However, this assumption may not always hold. To make the framework more robust, we enhance SPAR algorithm by incorporating feasibility cuts that ensure the feasibility of the obtained decision by not allowing the search to go into infeasible directions. Specifically, if  second stage problem is not feasible during an iteration, Farkas duals $\mu_\zeta$ corresponding to linking constraint~(\ref{eq:general second stage constraint I}) are obtained, and the feasibility cut of the form~(\ref{eq:feasibility_cut}) is added to the first stage problem~\cite{introduction_to_stochastic_programming}.

\useshortskip
\begin{equation}
     \mu_\zeta ^Tx \geq 0
    \label{eq:feasibility_cut}
\end{equation}


\begin{algorithm}[t]
    \caption{Separable Projective Approximation Routine}  \label{algo:SPAR general}
    \begin{spacing}{0.8}
    \begin{algorithmic}[1]
        \STATE Initialize: $k=1, m^{(1)}=\textbf{0}^L$ 
        \FOR {$k = 1,2,\dots,k^{\max}$}
            \STATE Get $l^{(k)} \in [\underline{L}, \overline{L}]$ after solving $\textbf{P}_{\textbf{g1}}$
            \STATE Observe scenario and get random variable $\gamma^{(k)}$ after solving $\textbf{P}_{\textbf{g2}}$ such that,\\ 
            $\mathbb{E}(\gamma^{(k)} | l^{(1)},\dots, l^{(k)}, \gamma^{(1)}, ..., \gamma^{(k-1)}) = \overline{m}_{l^{(k)}}$
            \STATE Update slope as:
            \begin{equation*}
                n_x^{(k)} = \begin{cases}
                    
                (1 - \alpha^{(k)}) m_x^{(k)} + \alpha^{(k)}\gamma^{(k)}\ \textbf{if}\ x = l^{(k)},\\
                m_x^{(k)}\ \textbf{if}\ x \neq l^{(k)}
                \end{cases}
            \end{equation*}

            \STATE Project $n^{(k)}\ \text{to convex set}\ \mathcal{M}\ \text{to yield}\ m^{(k+1)}$
            \STATE Check the termination criteria
        \ENDFOR
    \end{algorithmic}
    \end{spacing}
\end{algorithm}

\vspace{0.1em}
\subsection{Projection Algorithm}\label{sec:Projection algorithm}
In this section, we present a simple projection rule that meets optimality conditions while avoiding optimization, thereby significantly accelerating the algorithm~\cite{VFA_powell}.
The projection of updated slope vector $n$ on convex set $\mathcal{M}$ is:
\useshortskip
\begin{equation}
    \min_{m}\ \frac{1}{2}\|m-n\|  \label{eq: proj obj} 
    \end{equation}

\useshortskip


\useshortskip
\begin{equation}
  \text{\quad s.t. } \quad  m_l - m_{l+1} \leq 0 \quad \forall\ l \in \{\underline{L},\underline{L}+1, \dots, \overline{L}\} \label{eq: proj const}
\end{equation}
Taking $\beta_l$ as dual multipliers of (\ref{eq: proj const}), the optimality conditions for this problem can be written as, 


\useshortskip
\begin{equation} 
     m_l = n_l + \beta_{l-1} - \beta_{l} \quad \forall\ l \in \{\underline{L}+1,\underline{L}+2,\dots, \overline{L}\} \label{proj:opt I}
 \end{equation}

\useshortskip
 \begin{equation}
     \beta_l(m_l - m_{l+1}) = 0 \quad \forall\ l \in \{ \underline{L},\underline{L}+1, \dots, \overline{L}\} \label{proj:opt II} 
\end{equation}


Consider coordinates \(p\) and \(q\) with \(p \leq q\) such that,
\useshortskip
\begin{equation} \label{proj: equal equation}
m_{p-1} < m_p = ... = d = ... = m_q < m_{q+1}
\end{equation}

Adding (\ref{proj:opt I}) from \(l = p\) to \(l = q\) and using (\ref{proj: equal equation}), we obtain

\useshortskip
\begin{equation}
    \sum_{l = p}^{q} m_l = \sum_{l = p}^{q} n_l + \beta_{p-1} - \beta_q
    \end{equation}

\useshortskip
\begin{equation}
    (q - p + 1) d = \sum_{l = p}^{q} n_l + \beta_{p-1} - \beta_q \label{proj:expression}
\end{equation}

From (\ref{proj:opt II}) and (\ref{proj: equal equation}),\ \(\beta_{p-1} = 0 ,\beta_q = 0\), (\ref{proj:expression}) becomes,

\useshortskip
\begin{equation}\label{proj: d expression}
     d = \frac{1}{(q - p + 1)}\sum_{l = p}^{q} n_l
\end{equation}

In Algorithm~\ref{algo:SPAR general}, the slope vector is updated at a single coordinate, say $l^{(k)}$ at iteration $k$. Considering $n_{l}^{(k)}$ as a slope at point $l$ at iteration $k$, the updated slope at the current chosen point $l^{(k)}$ would be $n_{l^{(k)}}^{(k)}$. Since the slope is updated at only one point, the following three conditions can occur.

\begin{itemize}[leftmargin=*]
    \item \emph{Condition 1}: $n_{l^{(k)}-1}^{(k)} \leq n_{l^{(k)}}^{(k)} \leq n_{l^{(k)}+1}^{(k)}$\\
    The updated slope maintains the convexity. Hence, no further update is required.
    
    \item \emph{Condition 2}: $n_{l^{(k)}-1}^{(k)} > n_{l^{(k)}}^{(k)}$\\
    Find largest $\eta$ with $\underline{L} < \eta \leq l^{(k)}$ such that, 
    \useshortskip
    \[n_{\eta-1}^{(k)} \leq \frac{1}{l^{(k)} - \eta + 1} \sum_{l = \eta}^{l^{(k)}} n_l^{(k)}\]
    and assign, $m_l^{(k+1)} = d\quad \forall\ l \in \{\eta,~\eta +1,\dots,l^{(k)}\}$, where $d$ is obtained from (\ref{proj: d expression}) with $p = \eta\ \text{and}\ q = l^{(k)} $. If no $\eta$ is found, assign $\eta = \underline{L}$.
    
    \item \emph{Condition 3}: $n_{l^{(k)}}^{(k)} > n_{l^{(k)}+1}^{(k)}$\\
    Find smallest $\eta$ with with $l^{(k)} \leq \eta < \overline{L}$ such that,
    \useshortskip
    \[n_{\eta+1}^{(k)} \geq \frac{1}{\eta - l^{(k)} + 1} \sum_{l = l^{(k)}}^{\eta} n_l^{(k)}\]
    and assign, $m_l^{(k+1)} = d \quad \forall\ l \in \{l^k, l^k+1, \dots, \eta\}$ where $d$ is obtained from (\ref{proj: d expression}) with $p = l^k\ \text{and}\ q = \eta $. If no $\eta$ is found, assign $\eta = \overline{L}$.
    The Lagrange multipliers corresponding to the above solutions exactly matches with that of (\ref{proj:opt I}) and (\ref{proj:opt II})~\cite{VFA_powell}. 
\end{itemize}

        

\subsection{Statistical method for lower and upper bound}
Since the proposed framework uses sample subgradient information to approximate the function, we provide a statistical method for evaluating the quality of obtained solution through a posterior analysis, 
utilizing the information available during the process of solving the stochastic optimization problem.
We denote objective function $H(x,y) + \mathbb{E}_{\xi}[Q(x, \xi)]$ in~(\ref{eq:general opf first stage}) as  $F(x,y)$ and let $F^{*}$ be its optimal value. First, we provide the statistical upper bound, i.e., an estimate of the upper bound and its confidence interval at a chosen level. Let $(\hat{x},\hat{y})$ be the decisions obtained after solving (\ref{eq:general opf first stage})  where the function is approximated using SPAR-OPF taking batches of $N'$ samples from scenarios set $\Xi$. Consider that $p_\xi$ is the probability of scenario $\xi$ and $h_{|\Xi|}$ is the value of (\ref{eq:general opf first stage})  where expected function is calculated by simulating all scenarios independently in the second stage $\textbf{P}_\textbf{g2}$ fixing the first stage decision $(\hat{x}, \hat{y})$ as given in~(\ref{eq:UB epoch value}). This gives the upper bound for $F^*$. Taking expectation in the left side of~(\ref{eq:UB epoch value}) also gives the upper bound of $F^*$ as given in~(\ref{eq:UB expectation}) where (\ref{eq:general opf first stage}) is solved $M$ times taking batches of scenarios. The unbiased estimator, $\bar{h}_{|\Xi|}$, of $\mathbb{E}$($h_{|\Xi|}$) and variance $\sigma_M^2$ are given in (\ref{eq:UB mean}) and (\ref{eq:UB variance}), respectively. The confidence interval for the upper bound $UB_M$ with an approximated (1-$\alpha$) confidence is given in~(\ref{eq:UB CI}), which is justified by the central limit theorem with critical value $z_\alpha = f^{-1}(1-\alpha)$ where $f(z)$ is the cumulative distribution function of the standard normal distribution.  Note that if the value of $M$ is low, the $t$-distribution should be used.

\useshortskip
\begin{equation}
    h_{|\Xi|} =  c^T \hat{x} + d^T\hat{y} + \sum_{i\in 1}^{|\Xi|} p_{\xi}Q(\hat{x},\xi_i) \geq F^* \label{eq:UB epoch value}
\end{equation}

\useshortskip
\begin{equation}
 \mathbb{E}(h_{|\Xi|}) \geq F^{*} \label{eq:UB expectation}
\end{equation}

\useshortskip
\begin{equation}
    \bar{h}_{|\Xi|} =  \frac{1}{M}\sum_{i\in 1}^{M} h_{|\Xi|}^i \label{eq:UB mean}
\end{equation}

\useshortskip
\begin{equation}
     \sigma_M^2 = \frac{1}{M(M-1)}\sum_{i = 1}^{M} [\bar{h}_{|\Xi|} - h_{|\Xi|}^i]^2 \label{eq:UB variance}
\end{equation}

\useshortskip
 \begin{equation}
     UB_{M} = \bar{h}_{|\Xi|} \pm  z_{\alpha}\sigma_{M} \label{eq:UB CI}
\end{equation}

The lower bound for $F^*$ can be obtained by leveraging the sample average approximation. The sample average approximation of $F_{N^{'}}(x,y)$ with $N^{'}$ samples is an unbiased estimator of $F(x,y)$ as given in left equality in~(\ref{eq:LB relation_1}). If $p_{N^{'}}$ is an optimal value of (\ref{eq:general opf first stage}) for $N^{'}$ samples, then  $\mathbb{E}(p_{N^{'}})$ is a lower bound for $F(x,y)$ as given in~(\ref{eq:LB relation_1}). Taking minimum of $F(x,y)$ and using~(\ref{eq:LB relation_1}), we obtain (\ref{eq:LB relation_2}).
The unbiased estimator $\bar{p}_{N^{'}}$ of $\mathbb{E}(p_{N^{'}})$ and the variance are given in (~\ref{eq:LB mean}) and ($\ref{eq:LB variance}$), respectively, where problem (\ref{eq:general opf first stage}) is solved $M$ times taking $N^{'}$ samples. The confidence interval for the Lower bound  $LB_{N^{'}}$ of $F^*$ with an approximated (1-$\alpha$) confidence is given by~(\ref{eq:LB value}).



\useshortskip
\vspace{0.05em}
\begin{equation}
    F(x,y) = \mathbb{E}[{{F}_{N'}(x,y)}] \geq \mathbb{E}[{\min_{x,y} {F}_{N'}(x,y)]} = \mathbb{E}({p_{N'}}) \label{eq:LB relation_1}
\end{equation}

\useshortskip
\begin{equation}
    \min_{x,y} F(x,y) = F^* \geq \mathbb{E}{(p_{N'})} \label{eq:LB relation_2}
\end{equation}

\useshortskip
\begin{equation}
   \bar{p}_{N^{'}} =  \frac{1}{M}\sum_{i\in 1}^{M} p_{N^{'}}^{i} \label{eq:LB mean}
\end{equation}

\useshortskip
\begin{equation}
    \sigma_{N^{'}}^2 = \frac{1}{M(M-1)}\sum_{i = 1}^{M} [\bar{p}_{N^{'}} - p_{N^{'}}^{i}] \label{eq:LB variance}
\end{equation}

\useshortskip
\begin{equation}
   LB_{N^{'}} = \bar{p}_{N^{'}} \pm z_{\alpha}\sigma_{N^{'}} \label{eq:LB value}
\end{equation}


\begin{remark}
 Since SPAR-OPF solves the problem by approximating the value function, the approximation quality can be measured as the iterations proceed.  After approximating the function, either $\textbf{L}_\textbf{1}$ or $\textbf{E}_\textbf{1}$ is solved to get first stage decisions $(x^{(k)},y^{(k)})$, and then $H(x^{(k)},y^{(k)}) + \mathbb{E}_\xi [Q(x^{(k)},\xi)]$ is computed to measure quality of the approximated function.
 \end{remark}

\section{Stochastic DG Siting and Sizing Problem} \label{sec:DG siting as an example}
The proposed framework is implemented to solve the optimal DG siting and sizing problem in distribution systems. 
The distribution system is modeled as a graph $\mathcal{G(B,E)}$ where the node $i\in \mathcal{B}$ represent the bus and the edge $ij \in \mathcal{E}$ represents the line that connects nodes $i$ and $j$. We consider a three-phase unbalanced linear power flow model to represent the power flow physics for the distribution system~\cite{low2014convex}. The decisions variables and constraints associated with the two-stage problem are described below.

\subsection{First Stage} \label{sec:subsub:First stage}

 The objective function in the first stage problem $\textbf{D}\textbf{1}$ is given by (\ref{eq:first stage objective}), which minimizes the expected recourse cost $Q(x,\xi)$. The first stage involves identifying the DG siting decisions in each bus $i$, represented by the binary variable $\delta_i^{DG}$, and the corresponding number of DG units $U_i^{DG}$ to be installed adhering to the given budget. 
Constraint (\ref{eq: DG limits constraint}) ensures the size of installed DG is within the lower limit $\underline{P}^{DG}$ and upper limit $\overline{P}^{DG}$ where $u^{DG}$ is the rating of one DG unit.
Constraint (\ref{eq:Budget limit constraint}) ensures that the total cost of DG installation is within the budget $b$ where $c^{DG}_{pu}$ is the per kW cost of DG. Here, the cost of DG is assumed to vary linearly with DG size, and the fixed installation cost is not considered, which can be easily incorporated by just adding the constant term multiplied by the siting variable in the expression. Constraint~(\ref{eq: minim_maximum_DG_number bus constraint}) bounds the number of buses that DG can be installed between $\underline{N}^{DG}\ \text{and}\ \overline{N}^{DG}$. Contrary to prior works in~\cite{resilience_oriented_DG_siting_sizing} where DG locations were limited to keep the problem size small, we consider all buses as possible candidates for DG, which expands the feasibility space and enhances the quality of the solution. 


\useshortskip
\begin{equation}
    (\textbf{D}\textbf{1}):\quad  \min\ \mathbb{E}_\xi [Q(x,\xi)] \label{eq:first stage objective}
\end{equation}


\useshortskip
\begin{equation}
  \text{\quad\quad s.t. }   \underline{P}^{DG} \delta_i^{DG} \leq U_i^{DG} u^{DG} \leq \overline{P}^{DG} \delta_i^{DG} \quad \forall\ i \in \mathcal{B} \label{eq: DG limits constraint}
\end{equation}

\useshortskip
\begin{equation}
     \sum_{i=1}^{|\mathcal{B}|} U_i^{DG}u^{DG}c^{DG}_{pu} \leq b \label{eq:Budget limit constraint}
\end{equation}

\useshortskip
\begin{equation}
    \underline{N}^{DG} \leq \sum_{i=1}^{|\mathcal{B}|} \delta_i^{DG} \leq \overline{N}^{DG} \label{eq: minim_maximum_DG_number bus constraint}
\end{equation}

\subsection{Second Stage}\label{sec:sub:second stage problem}


The second stage OPF problem minimizes the voltage deviation in each node for a range of scenarios. Any parameter or variable associated with the scenario $\xi \in \Xi$ and phase $\phi \in \Phi$ is superscripted by $\xi$ and $\phi$, respectively.
The objective function is given in~(\ref{eq:second stage objective I}) where $v_i$ and $v_{i,ref}$ represent the observed voltage and reference voltage of bus $i$, respectively. 

\begin{equation}
        (\textbf{D}_\textbf{2}^\textbf{v}):  Q(x, \xi) = \min \sum_{i\in \mathcal{B}} \sum_{\phi \in \Phi}|v_i^{\phi,\xi} - v_{i,ref}^{\phi}| \label{eq:second stage objective I}
\end{equation}

Since the objective function~(\ref{eq:second stage objective I}) is non-smooth, it is reformulated as~(\ref{eq:second stage objective II}) by introducing extra variable $z_i^{\phi,\xi}$ and~(\ref{eq: objective reformulation constraint}), which ensures $z_i^{\phi,\xi}$ = $|v_i^{\phi,\xi} - v_{i,ref}^\phi|$.

\begin{equation}
    Q(x, \xi) = \min \sum_{i\in \mathcal{B}} \sum_{\phi \in \Phi} z_i^{\phi,\xi} \label{eq:second stage objective II}
\end{equation}

\begin{equation}
    -z_i^{\phi,\xi} \leq\ v_i^{\phi,\xi} - v_{i,ref}^{\phi} \leq z_i^{\phi,\xi}\quad  \forall\ i \in \mathcal{B},  \quad \forall\ \phi \in \Phi \label{eq: objective reformulation constraint}
\end{equation}




Constraint (\ref{eq: Linking constraint}) is the linking constraint between the first and second stages ensuring that the DG power $U_i^{DG, \phi,\xi}$ utilized in the given scenario is within the DG capacity decision made in the first stage. 
Constraints~(\ref{eq:active power balance constraint})--(\ref{eq:voltage bound constraint}) represent the three-phase unbalanced linearized power flow equations~\cite{low2014convex}. 
Active power balance for each bus is maintained by constraint~(\ref{eq:active power balance constraint}) where   $P_{i,load}^{\phi}$ denotes base active load of bus $i$ for phase $\phi$. $\epsilon^\xi_{i,PV}$ and $\epsilon^\xi_{i,load}$ represent uncertainties in PV and load, respectively. $\mathcal{P}_i$ and $\mathcal{C}_i$ are sets of parent and child buses of $i$ while $P_{ij}^\phi$ and $i_{sub}$ denote the active power flowing from bus $i$ to $j$ and substation bus index respectively. Similarly, reactive power balance at the given bus is given by~(\ref{eq:reactive power balance constraint}) where $Q$ represents the reactive power and $\theta^{DG}$ is the power factor angle of DG. Let $\textbf{v}_i^{\xi}$ denote the three-phase voltage square magnitude of bus $i$ while $\textbf{P}_{ij}^{\xi}$ and $\textbf{Q}_{ij}^{\xi}$  are three-phase active and reactive power flowing through line $ij$, respectively. Constraint~(\ref{voltage drop constraint}) is the voltage balance equation of line $ij$. $\textbf{r}_{ij}$ and $\textbf{x}_{ij}$ represent the resistance and reactance matrix of line $ij$, respectively. Constraint~(\ref{eq:voltage bound constraint}) bounds bus voltage square $\textbf{v}_i^{\xi}$ between minimum $\textbf{v}_{\min}$ and maximum $\textbf{v}_{\max}$ allowed values.

\begin{equation}
     \sum_{\phi \in \Phi}U_i^{DG,\phi,\xi} \leq U_i^{DG}u^{DG} \quad \forall\ i \in \mathcal{B} \label{eq: Linking constraint}
    \end{equation}

\useshortskip
\begin{equation}
\begin{gathered}
     U_i^{DG, \phi,\xi}\epsilon_{i,PV}^{\xi}- P_{i,{\rm load}}^{\phi}\epsilon_{i,{\rm load}}^\xi = \sum_{j \in \mathcal{C}_i} P_{ij}^{\phi,\xi}
     - \sum_{j \in \mathcal{P}_i} P_{ji}^{\phi,\xi} \\ \quad \forall\ i \in \mathcal{B}\backslash \{i_{sub}\},  \quad \forall\ \phi \in \Phi \label{eq:active power balance constraint}
\end{gathered}
\end{equation}

\begin{equation}
\begin{gathered}
     U_i^{\phi,\xi}\epsilon_{i,PV}^{\xi}\tan \theta^{DG} - Q_{i,{\rm load}}^{\phi,\xi}\epsilon_{i,{\rm load}}^\xi = \sum_{j \in \mathcal{C}_i}  Q_{ij}^{\phi,\xi}
    -\\ \sum_{j \in \mathcal{P}_i} Q_{ji}^{\phi,\xi} \quad
    \forall\ i \in \mathcal{B}\backslash \{i_{sub}\},\quad \forall\ \phi \in \Phi \label{eq:reactive power balance constraint}
\end{gathered}
\end{equation}


\useshortskip
\begin{equation}
     \textbf{v}_i^{\xi} - \textbf{v}_j^{\xi} = 2(\tilde{\textbf{r}}_{ij}\textbf{P}_{ij}^{\xi} + \tilde{\textbf{x}}_{ij}\textbf{Q}_{ij}^{\xi}) \quad \forall\ ij \in \mathcal{E} \label{voltage drop constraint}
\end{equation}

\useshortskip
\begin{equation*} 
\text{where, }~    \tilde{\textbf{r}}_{ij} = {\rm Real}\{\alpha \alpha^H\}\otimes \textbf{r}_{ij} + {\rm Im}\{\alpha \alpha^H\}\otimes \textbf{x}_{ij} \hspace{0.5in}\label{eq:rij related expression}
\end{equation*}

\useshortskip
\begin{equation*}
     \tilde{\textbf{x}}_{ij} = {\rm Real}\{\alpha \alpha^H\}\otimes \textbf{x}_{ij} + {\rm Im}\{\alpha \alpha^H\}\otimes \textbf{r}_{ij} \label{eq: xij related expression}
\end{equation*}

\useshortskip

\useshortskip
\begin{equation*}
     \alpha = [1\; e^{-j2\pi/3}\; e^{j2\pi/3}]^T
\end{equation*}

\useshortskip
\begin{equation}
     \textbf{v}_{\min} \leq \textbf{v}_i^{\xi} \leq \textbf{v}_{\max} \quad  \forall\ i \in \mathcal{B} \label{eq:voltage bound constraint}
\end{equation}

Constraint~(\ref{eq: line flow constraint}) ensures that power flowing through line $ij$  is within its thermal rating $\textbf{S}^{\rm rated}_{ij}$. This quadratic constraint is substituted with linear constraints (\ref{eq:line flow constraint I}), (\ref{eq:line flow constraint II}) and (\ref{eq:line flow constraint III}) using polygon based approximations~\cite{resilience_oriented_DG_siting_sizing}. 

\useshortskip
\begin{equation}
     \textbf{P}_{ij}^{\xi^2} + \textbf{Q}_{ij}^{\xi^2} \leq (\textbf{S}_{ij}^{\rm rated})^2 \quad \forall\ ij \in \mathcal{E} \label{eq: line flow constraint}
\end{equation}

\useshortskip
\begin{equation}
     -\sqrt{3}(\textbf{P}_{ij}^\xi + \textbf{S}_{ij}) \leq \textbf{Q}_{ij}^\xi \leq -\sqrt{3}(\textbf{P}_{ij}^\xi - \textbf{S}_{ij}) \quad \forall\ ij \in \mathcal{E} \label{eq:line flow constraint I}
\end{equation}

\useshortskip
\begin{equation}
    -\frac{\sqrt{3}}{2}\textbf{S}_{ij} \leq \textbf{Q}_{ij}^\xi \leq \frac{\sqrt{3}}{2}\textbf{S}_{ij} \quad \forall\ ij \in \mathcal{E} \label{eq:line flow constraint II}
\end{equation}

\useshortskip
\begin{equation}
     \sqrt{3}(\textbf{P}_{ij}^\xi - \textbf{S}_{ij}) \leq \textbf{Q}_{ij}^\xi \leq \sqrt{3}(\textbf{P}_{ij}^\xi + \textbf{S}_{ij}) \quad \forall\ ij \in \mathcal{E} \label{eq:line flow constraint III}
\end{equation}

\useshortskip

\useshortskip
\begin{equation*}
\text{where }~     \textbf{S}_{ij} = \textbf{S}_{ij}^{\rm rated}\sqrt{\frac{2\pi/6}{\sin(2\pi/6)}} \,. \hspace*{0.3in}
\end{equation*}

\begin{remark}
    $\textbf{D}_\textbf{2}^\textbf{v}$ can be adapted to minimize power loss by replacing objective function in~(\ref{eq:second stage objective II})  with a quadratic approximated power loss objective and removing (\ref{eq: objective reformulation constraint}); the resulting $\textbf{D}_\textbf{2}^\textbf{p}$ is given in~(\ref{eq:second stage objective power loss}). Voltage $v_i^{\phi,\xi}$ in the objective term is assumed to be nominal voltage $v_{i,nom}^{\phi}$~\cite{dubey_monograph}. We also analyzed SPAR-OPF with $\textbf{D}_\textbf{2}^\textbf{p}$, which incorporates a quadratic function in the second stage.
\end{remark}

\useshortskip
\begin{equation}
    (\textbf{D}_\textbf{2}^\textbf{p}):\quad Q(x, \xi) =  \min \sum_{ij\in \mathcal{E}} \sum_{\phi \in \Phi} \frac{(P_{ij}^{\phi,\xi})^2 + (Q_{ij}^{\phi,\xi})^2}{v_i^{\phi,\xi}}r_{ij} \label{eq:second stage objective power loss} 
\end{equation}
\begin{equation}
   \text{\quad s.t. } \quad (\ref{eq: Linking constraint})\text{--}(\ref{eq:line flow constraint III})
\end{equation}


\subsection{Load and PV Uncertainty Model}\label{subsec:load and PV uncertainty}
We consider the temporal relationship between load and PV outputs based on historical observations of load and PV profiles at a fixed time resolution.
For the PV output model, we leverage the PVwatts tool developed by the National Renewable Energy Laboratory (NREL)~\cite{PV_watt_manual}. The irradiance and temperature of the location are considered to obtain a normalized PV output, also called PV multiplier. Since our main focus is not on scenario generation and reduction, a simple but realistic model motivated by stratified and importance sampling is used in this work~\cite{risk_based_paper_abodh}. 

\begin{algorithm}
    \caption{Scenario Generation and Reduction Algorithm}\label{algo: scenario generation and reduction}
    \begin{spacing}{0.6}
    \begin{algorithmic}[1]
        \STATE \textbf{Input:} $\hat{\Xi}, N_t  \text{ and } \textbf{Output: } \Xi = \{\xi_1, \xi_2, ..., \xi_{N_t}\}$
        \STATE $N_D^{\rm agg} \gets \frac{|\hat{\Xi}|}{N_t}$

        \STATE initialize $s = 1, b = 1, e = N_D^{\rm agg} $
        \WHILE{$e \leq |\hat{\Xi}|$}
            \FOR {$h = 1:24$}
                \STATE $\alpha^{\rm load}(s) = \mathbb{E} [\alpha_h^{\rm load}(d)]_{d=b}^e$\\
                \STATE $\alpha^{PV}(s) = \mathbb{E} [\alpha_h^{PV}(d)]_{d=b}^e$
                \STATE $s \gets s+1$
            \ENDFOR
        
            \STATE $b \gets b + N_D^{\rm agg}, \quad e \gets e + N_D^{\rm agg}$

        \ENDWHILE
        
        \FOR {$s = 1:N_t$}
           \STATE $\alpha_i^{\rm load}(s) =  \alpha^{\rm load}(s) + \delta({\alpha_i^{\rm load}}(s)) \quad \forall\ i \in \mathcal{B}$ \label{stepload}
           \STATE $\alpha_i^{PV}(s) =  \alpha^{\rm load}(s) + \delta({\alpha_i^{PV}}(s))\ \quad \forall\ i \in \mathcal{B}$ \label{stepPV}
            
        \ENDFOR
        \STATE \textbf{Return:} $\Xi\ \text{where}\ \xi_s = [\alpha_i^{\rm load}(s), \alpha_i^{PV}(s) \quad \forall\ i \in \mathcal{B}]$
        
    \end{algorithmic}
    \end{spacing}
\end{algorithm}

The algorithm for scenario generation and reduction is presented in Algorithm~\ref{algo: scenario generation and reduction}, which takes  $\hat{\Xi}$ scenarios consisting of hourly resolution historical normalized load and PV output and provides $N_t$ output scenarios. The scenario index of scenario $\xi_s \in \Xi$ is denoted by $s$ in the algorithm. The number of days to be aggregated $N_D^{\rm agg}$ in a stratum is calculated based on total historical data available and the number of scenarios desired $N_t$. Then, aggregation is done for each hour for these aggregated days, the index of which starts from $b$ to $e$, to obtain scenarios. $\alpha_h^{\rm load} (d) \in \hat{\Xi}$  and $\alpha_h^{PV} (d) \in \hat{\Xi}$ denote normalized load and PV multiplier at hour $h$ of day $d$, respectively. $\mathbb{E} [\alpha_h^{\rm load}(d)]_{d=b}^e$ and $ \mathbb{E} [\alpha_h^{PV}(d)]_{d=b}^e$ denote the expectation of load and PV multiplier, respectively, corresponding to hour $h$ considering days indexed from $b$ to $e$. Then, once all 24 hours are done, the day's index will be sided by $N_D^{\rm agg}$ to obtain scenarios considering the next strata of aggregated days and so on until the last data point in $\hat{\Xi}$ is reached. 

This gives the uniform load $\alpha^{\rm load}(s)$, and PV $\alpha^{PV}(s)$, multiplier profile for each scenario $\xi_s$ for the given system. In reality, certain variations in load and PV profile exist among buses. Thus, the load and PV multiplier for each bus is made different as shown in steps \ref{stepload} and \ref{stepPV}, respectively, where $\delta({\alpha_i^{\rm load}}(s))$ and $\delta({\alpha_i^{PV}}(s))$ represent the small deviations from the assumptions of uniform shape for load and PV, respectively, for bus $i$ for scenario  $\xi_s$.
These deviations are generated from Gaussian random noise. Each process is assumed to be independent of the other and identically distributed.
The user can choose the standard deviation value depending on the diversity of historical profiles. Finally, the algorithm provides a scenario set $\Xi$, where each scenario contains load and PV multiplier corresponding to each bus. 






\subsection{SPAR-OPF algorithm for DG siting and sizing}

Algorithm~\ref{algo: SPAR for power systems planning} shows the overall algorithm using SPAR-OPF for DG siting and sizing problem which takes the power flow model, scenarios set, and other required parameters used in the first and second stage models
as input and outputs the planning decisions: DG size $U_i^{DG^{\text{final}}}$ and siting decisions $\delta^{DG^{\text{final}}}$. The number of coordinate functions to be learned equals the number of buses, as DG decisions for each bus are linked to the second stage and we consider all buses as potential locations. Consider L as the maximum number of DG units that can be installed in a bus. Let $m_i^{l,(k)} \in m_i^{(k)}$ denote the slope of the coordinate function $g_i$ at point $l\in \mathcal{L} = [0,L]$ during iteration $k$ where $m_i^{(k)}$ is its slope vector.  The slope $m_i^{l,(1)}$ of all coordinate functions is initialized to zero. The maximum number of iterations is set as $k^{\max}$. 
Then, at iteration $k$, with  $m_i^{(k)}$ and using~(\ref{eq:learned function value SPAR}), the value of $g_i$ is updated. $g_i^{l,(k)}$ denotes its value at $l$ in the $k^{th}$ iteration. The first stage problem is then solved using either $\textbf{L}_1$ or $\textbf{E}_1$ and DG decisions (linking variables) $U_i^{DG^{(k)}}$ are passed to the second stage. Scenario $\xi$ is selected randomly from the scenario set $\Xi$ and uncertain parameters specified in~(\ref{eq:active power balance constraint}) and (\ref{eq:reactive power balance constraint}) in the second stage problem~\ref{sec:sub:second stage problem} are updated. Then, the second stage problem is solved, and if it is feasible, the dual information $\gamma_i^{(k)}$ of the linking constraints (\ref{eq: Linking constraint}) is obtained.  The slope of each coordinate function at the given first stage decision $U_i^{DG^{(k)}}$ is then adjusted using $\gamma_i^{(k)}$ where updated step size $\alpha^{(k)}$ is utilized. 
Using the projection algorithm as described in Section \ref{sec:Projection algorithm}, updated slope vector $n_i^{(k)}$ is then projected to the convex set to obtain $m_i^{(k+1)}$, which is utilized in the next iteration. On the other hand, if the second stage is infeasible, a feasibility cut of the form~(\ref{eq:feasibility_cut}) is added to the first-stage problem. The process continues until the termination criterion is met.

\begin{algorithm}[tbhp]
    \caption{SPAR-OPF for DG Siting and Sizing}\label{algo: SPAR for power systems planning}
    \begin{spacing}{0.8}
    \begin{algorithmic}[1]
        \STATE \textbf{Input:} $G, \Xi, k^{\max}, \mathcal{L}, \underline{P}^{DG}, \overline{P}^{DG}, c^{DG}_{pu}, u^{DG}, b, \theta^{DG}, \textbf{x}_{ij}, \textbf{r}_{ij}$
        \STATE \textbf{Output:} $U_i^{DG^{\text{final}}},\delta_i^{DG^{\text{final}}}$
        \STATE $m_i^{l,(1)} \gets 0 \quad \forall\ i \in \mathcal{B},\  \forall\ l \in \mathcal{L}$  
        
        \FOR {$k = 1,2,\dots,k^{\max}$}
            \STATE Update $g_i^{l,(k)}\  \text{using}\ m_i^{(k)}$
            \STATE $U_i^{DG^{(k)}} \gets $ Solve first stage problem  (\ref{sec:subsub:First stage})
            \STATE Draw random scenario $\xi \in \Xi$
            \FOR {each bus $i \in \mathcal{B}$}
                \STATE update $ \epsilon_{i,PV}^\xi ,\epsilon_{i,{\rm load}}^\xi$ in (\ref{eq:active power balance constraint}) and (\ref{eq:reactive power balance constraint})
                \STATE update $U_i^{DG} \gets U_i^{DG^{(k)}}$ in (\ref{eq: Linking constraint})
            \ENDFOR
            \STATE Solve second stage problem
            (\ref{sec:sub:second stage problem})
            \IF {second stage is feasible}
                \STATE $\gamma_i^{(k)} \gets \text{Constraint} (\ref{eq: Linking constraint}$)  \COMMENT {Lagrange Multipliers}
                \STATE Update $\alpha^{(k)}$
                \FOR {$ \text{each bus}\ i \in \mathcal{B}$}
                    \FOR {$l = 0,1,2,\dots, L$}
                        \IF {$l = U_i^{DG^{(k)}}$}
                            \STATE $n_i^{l,(k)} = (1-\alpha^{(k)})m_i^{l,(k)} + \alpha^{(k)}\gamma_i^{(k)}$
                        \ELSE
                            \STATE $n_i^{l,(k)} = m_i^{l,(k)}$
                        \ENDIF
                        \STATE $m_i^{(k+1)} \gets n_i^{(k)}\ \text{using projection algorithm \ref{sec:Projection algorithm}}$
                    \ENDFOR
                \ENDFOR
            \ELSE 
                    \STATE Add feasibility cut~(\ref{eq:feasibility_cut}) to first stage problem (\ref{sec:subsub:First stage})
            \ENDIF
        \ENDFOR
        \STATE Check the termination criteria
        \STATE \textbf{Return:} $U_i^{DG^{\text{final}}} \gets U_i^{DG^{(k)}}, \delta_i^{DG^{\text{final}}} \gets \delta_i^{DG^{(k)}}$
        
    \end{algorithmic}
    \end{spacing}
\end{algorithm}





\section{Results and Discussions} \label{sec:Results}

The effectiveness and scalability of SPAR-OPF are illustrated using two test systems: IEEE 123-bus and a 9500-node (modified IEEE 8500-node) distribution test system.
We compare solutions from SPAR-OPF with solutions from EF, 
which provides a measure of optimality.
The solution obtained from the proposed framework is also compared with that from PH, which uses scenario-based decomposition. The heuristic parameter $\rho$ for PH is set to 1. We refer the reader to~\cite{progressive_hedging_latest} for additional discussion on PH.

 The default learning step size rule for SPAR-OPF, which we denote here as ``step rule 1" is taken as $ \alpha^{(k)} = 20/(20+k)$ as suggested by~\cite{VFA_powell}.  Furthermore, different step size rules are tested to measure the robustness of the framework with respect to step size selection. Specifically, we define ``step size rule 2" as  $\alpha^{(k)} = 1/k$, and ``step size rule 3" as $\alpha^{(k)} = min(1,20/k)$ to cover different varieties. $k^{\max}$ and $\epsilon-radius$ are set to 100 and $10^{-4}$. The results, including the effects of increasing number of scenarios and system size, are then presented along with statistical validation. Finally, the computational comparison between the two formulations $\textbf{L}_1$ and $\textbf{E}_1$ is made. In the remainder of this paper, $\textbf{D}_2^\textbf{v}$ refers to two-stage stochastic optimization problem with voltage deviation minimization as an objective in the second stage, while $\textbf{D}_2^\textbf{p}$ refers to the one with power loss minimization as an objective.   



\subsection{Simulation Set-Up}

We use Pyomo~\cite{bynum2021pyomo} to model SPAR-OPF, PySP package in Pyomo to model EF and PH, and Gurobi to solve all models~\cite{gurobi}. All the codes and data implemented to develop SPAR-OPF, EF and PH will be made publicly available once the paper is accepted. 
All simulations are carried out on a workstation with 32 GB RAM and an Intel Core i9-10900X CPU @ 3.70GHz. 

For simulations, we consider $\overline{N}^{DG} = 10$, $\underline{P}^{DG} = 33 \text{kW}$, $\overline{P}^{DG} = 333 \text{kW}$, $u^{DG} = 2 \text{kW}$, and $c_{pu}^{DG} = \$1.01$~\cite{Solar_cost}. A total budget $b$ of \$1500k is assumed. Note that all these parameters can be easily relaxed based on the user's preference and available resources. All results are generated by averaging the 25 independent random seeds. Unless otherwise mentioned, SPAR-OPF utilized the epigraph $\textbf{E}_1$ formulation. 

Scenarios are generated and reduced based on Algorithm~\ref{algo: scenario generation and reduction} described in section~\ref{subsec:load and PV uncertainty}. We add Gaussian noise with zero mean and 10\% standard deviation on load and PV profile to obtain bus profiles.
 Hourly resolution weather data is obtained from the National Solar Radiation Data Base (NSRDB) considering El Paso, the western area of Texas, which is served by the Electric Reliability Council of Texas (ERCOT) for the year 2023 \cite{NSRDB}. Hourly resolution historical load data for the year 2023 of the same area is taken from ERCOT~\cite{ERCOT}. This gives 8760 scenarios representing each hour of the year. Then, 96 and 1200 scenarios are generated for each system using Algorithm~\ref{algo: scenario generation and reduction}.  For instance, Fig.~\ref{fig:123 bus sytem different load multiplier} and Fig.~\ref{fig:123 bus sytem different PV multiplier} show the load and PV profile of 96 scenarios for IEEE 123-bus system, respectively, where solid line denotes the common shape and shaded curves are profiles for each bus. 


\begin{figure}[tbp]
\centering
  \begin{subfigure}{0.5\linewidth}
  \centering
    \includegraphics[width=\linewidth]{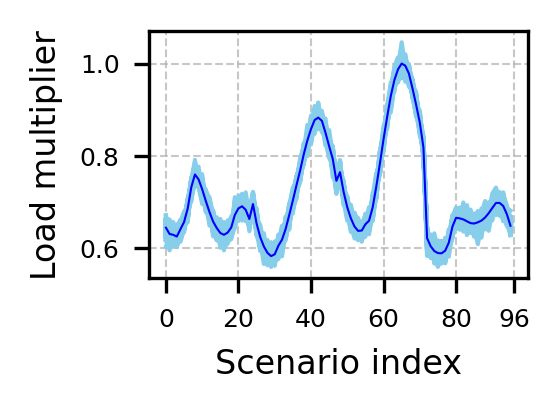}
    \caption{Bus load profile}
    \label{fig:123 bus sytem different load multiplier}
  \end{subfigure}%
  \begin{subfigure}{0.5\linewidth}
    \centering
    \includegraphics[width=\linewidth]{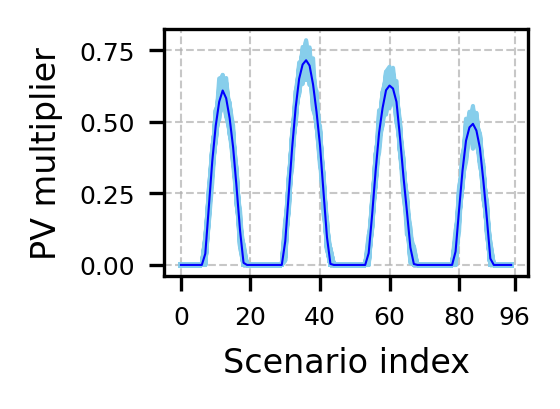}
    \caption{Bus PV profile}
    \label{fig:123 bus sytem different PV multiplier}
  \end{subfigure}
  \caption{Bus load and PV profile for IEEE 123 bus system.}
\end{figure}

\subsection{Illustration of SPAR-OPF}
In this section, we illustrate the performance of SPAR-OPF  on the IEEE 123 bus system with 96 and 1200 scenarios. The total demand is 3431.3 kW and 1931.9 kvar. We consider a system base of 1 MVA and voltage base of 4.16 kV.
\subsubsection{Case I: With 96 scenarios}

 
 First, we present a visualization that illustrates the learning of the coordinate functions by SPAR-OPF.  Recall that each coordinate function represents value as a function of the corresponding linking variable. Here, we learn a total of  123 coordinate functions, as DG sizing decisions (corresponding to each bus) are linked to the second stage. Fig.~\ref{fig:123 bus sytem coordinate learn functions} shows 9 such coordinate functions for the first five iterations with dark color showing the later iterations for problem $\textbf{D}_2^\textbf{v}$. The x-axis shows the decision points, and the y-axis shows the value of the function. It can be observed that in every iteration, the function updates its value due to changes in slope information from the algorithm. For the coordinate function $g_1$, the function remains constant at every decision point, i.e., equal to the initialized zero value. This function is associated with the substation bus, and installing DG at the substation bus does not add any value to the objective function considered here.
\begin{figure}[t]
    \centering
    \includegraphics[width=0.82\linewidth]{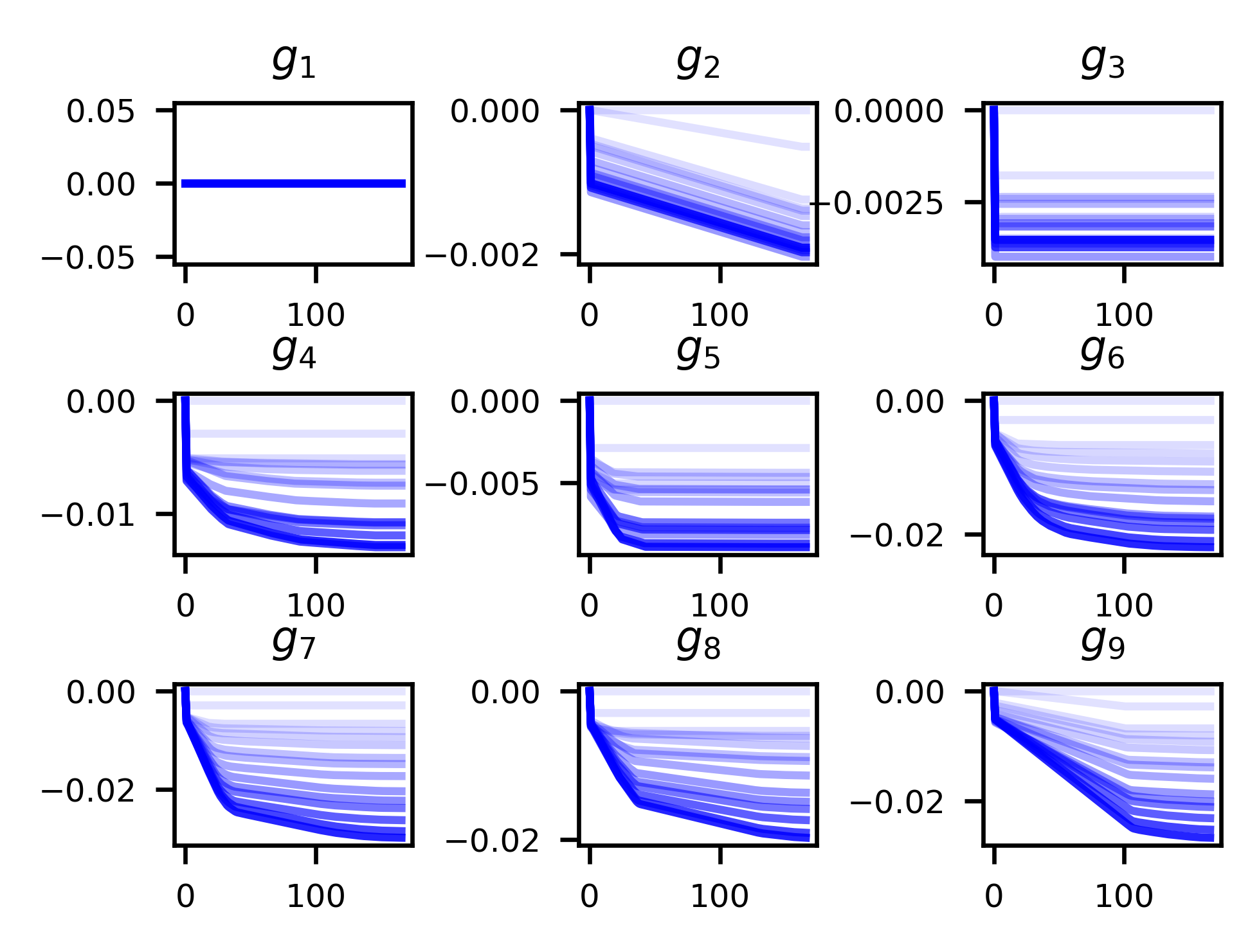}
    \caption{Illustration of coordinate functions learning. 
    } 
    \label{fig:123 bus sytem coordinate learn functions}
\end{figure}
Fig.~\ref{fig:quality for voltage expected function plot 123 bus system} shows the quality of function approximations with respect to iterations for $\textbf{D}_2^\textbf{v}$.
It can be seen that all step size rules give satisfactory performance with a quality of more than 98\%. The solution quality is consistent for $\textbf{D}_{2}^\textbf{p}$ as well, as shown in Fig.~\ref{fig:quality for power loss expected function plot 123 bus system}.

\begin{figure}[b]
\centering
  \begin{subfigure}{0.5\linewidth}
  \centering
    \includegraphics[width=\linewidth]{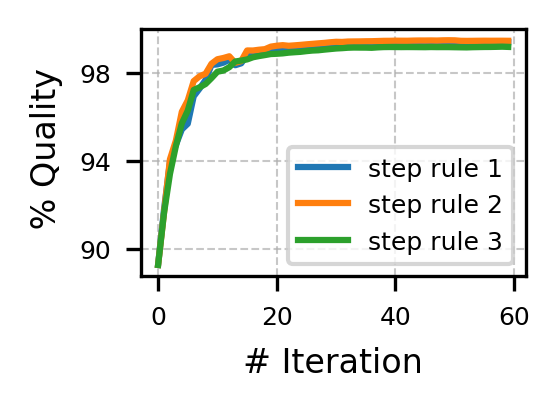}
    \caption{$\textbf{D}_\textbf{2}^\textbf{v}$}
    \label{fig:quality for voltage expected function plot 123 bus system}
  \end{subfigure}%
  \begin{subfigure}{0.5\linewidth}
    \centering
    \includegraphics[width=\linewidth]{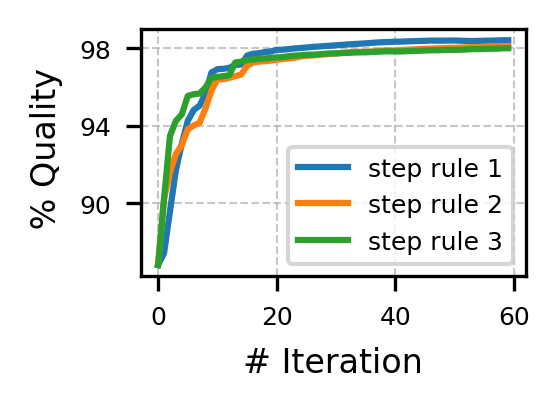}
    \caption{$\textbf{D}_\textbf{2}^\textbf{p}$}
    \label{fig:quality for power loss expected function plot 123 bus system}
  \end{subfigure}
  \caption{Quality of function approximation with SPAR-OPF.}
\end{figure}

Table~\ref{tab:IEEE 123 with 96 scenarios} presents the results from SPAR-OPF, EF, and PH where columns ``Obj" and ``Time" indicate the objective value and time taken  to solve the problem, respectively. It is important to mention that all algorithms were implemented on same computing environment. 
 For $\textbf{D}_2^\textbf{v}$, it is observed that SPAR-OPF objective value is very close to the true objective value (EF column) with a gap of 0.09 (0.44\%). The last row shows solutions for $\textbf{D}_2^\textbf{p}$ where the solution obtained from SPAR-OPF is very close to that of optimal value from EF with a gap of 0.6 kW (0.93\%). Note that the proposed framework gives a better solution than PH for $\textbf{D}_2^\textbf{v}$, whereas later performs better for $\textbf{D}_2^\textbf{p}$ for the considered problem; nevertheless, both solutions are close to the respective optimal values. 

\begin{table}[t]
\small
\centering
\setlength{\tabcolsep}{3pt}
\caption{IEEE 123 system with  96 scenarios}
\label{tab:IEEE 123 with 96 scenarios}
\begin{tabularx}{\linewidth}{ccccccc} \toprule
Problem & \multicolumn{2}{c}{\textbf{SPAR-OPF}} & \multicolumn{2}{c}{\textbf{EF}} & \multicolumn{2}{c}{\textbf{PH}} \\
 & Obj & Time(s) & Obj & Time(s) & Obj & Time(s) \\ \midrule
$\textbf{D}_\textbf{2}^\textbf{v}$ & 20.32 & 53.6 & 20.23 & 47.94 & 21.98 & 537 \\
$\textbf{D}_\textbf{2}^\textbf{p}$ & 64.7 kW & 58.8 & 64.1 kW & 686 & 64.24 kW & 6767\\ \bottomrule
\end{tabularx}
\raggedright \footnotesize{\noindent Note: The objective values for $\textbf{D}_\textbf{2}^\textbf{v}$ are unitless.}
\end{table}

SPAR-OPF and EF took comparable time to solve $\textbf{D}_2^\textbf{v}$, but PH took longer to solve. For the quadratic objective problem $\textbf{D}_2^\textbf{p}$, EF took 11.6x more time than SPAR-OPF. Due to the quadratic objective in the second stage, EF involves solving a huge MINLP problem, whereas SPAR-OPF solves smaller MILP in the first stage and QP in the second stage, thereby decreasing the overall solve time. Due to this reason, SPAR-OPF took comparable time to solve $\textbf{D}_2^\textbf{v}$ and $\textbf{D}_2^\textbf{p}$, whereas time for EF increases by 14x to solve the later. This is another advantage of the proposed framework, which decomposes the MINLP problem into easier sub-problems. PH takes a comparatively longer time to solve both problems. This performance time might get lower if PH had been implemented in a high-performance computing device. However, PH also involves  MINLP subproblems, which would take more time even if scenario subproblems are solved in parallel.

\subsubsection{Case II: with 1200 scenarios}

In this section, we illustrate statistical bounds and compare solution quality with EF and PH using a significantly larger number of scenarios, and analyze computational time.
Table~\ref{tab:IEEE 123 system with 1200 scenarios for voltage deviation} presents the results for $\textbf{D}_\textbf{2}^\textbf{v}$ considering 1200 scenarios. 
It is observed that EF for this case has 5.28 millions rows and 2.33 millions columns whereas case I has 0.42 millions rows and 0.186 millions columns, which illustrates the increasing complexity of the problem with an increase in the number of scenarios. 
The column $N_s$ in Table~\ref{tab:IEEE 123 system with 1200 scenarios for voltage deviation}  indicates the batch size taken for each random seed while solving the problem using SPAR-OPF. The columns ``Obj", ``LB CI", and ``UB CI" refer to objective value, lower bounds confidence interval and upper bound confidence interval for SPAR-OPF respectively. Similarly, column ``BG" refers to the bounds gap with \mbox{$(1-2\alpha)$}, where $\alpha = 0.05$, confidence level calculated based on the difference between the upper value of the upper bound confidence interval and the lower value of the lower bound confidence interval. The number within the bracket shows the percentage bounds gap to the objective value. Note that the objective value from SPAR-OPF lies within the lower bound and upper bound confidence interval. For smaller batch sizes, the width of the confidence interval and bounds gap are higher as the batch size is insufficient to capture the uncertainty set. The proposed framework obtains a solution closer, and better than the PH solution, to the optimal solution obtained from the EF.
The width of UB CI does not change that much as all scenarios are simulated, and less variability is observed.

It is important to mention that a bounds gap of $x\%$ does not necessarily mean that the objective value of SPAR-OPF is $x\%$ away from the optimal solution. 
It measures the tightness of bounds at the current solution. Take an instance, for case with $N_s = 200$, the objective value from SPAR-OPF is only 0.05\%, far from the optimal value though the BG\% is 2.1\%. It is observed that the proposed framework took 67.8 seconds per seed, whereas EF and PH took 1.86 hours and 2.89 hours to solve the problem, respectively.

Table~\ref{tab:IEEE 123 with 1200 scenario for power losss} presents results for problem $\textbf{D}_\textbf{2}^\textbf{p}$ with varying batch size where bounds gap and optimality gap are less than 4\% and 0.62\% respectively for batch size $N_s$ beyond 150.
EF, which involved solving a large MINLP problem, took 16.8 hours to complete, while PH, requiring the solution of 1200 MINLP problems per iteration, took 22.4 hours. On the other hand, SPAR-OPF took 79.6 seconds per seed 
to solve the same problem.
This shows the benefit of SPAR-OPF, which takes one scenario at a time, making it practical even for environments with limited resources.
\begin{table}[t]
\setlength{\tabcolsep}{3.5pt}
\small
\centering
\caption{IEEE 123 system with 1200 scenarios for $\textbf{D}_\textbf{2}^\textbf{v}$}
\label{tab:IEEE 123 system with 1200 scenarios for voltage deviation}
\begin{tabularx}{\columnwidth}{ccccccc}
\toprule
\multicolumn{1}{c}{$N_s$} & \multicolumn{1}{c}{Obj} & \multicolumn{1}{c}{LB CI} & \multicolumn{1}{c}{UB CI}  & \multicolumn{1}{c}{BG (\%)} & \multicolumn{1}{c}{EF} & \multicolumn{1}{c}{PH} \\ \midrule
50 & 16.99 & [16.67-17.31] & [17.41-17.53]& 0.86 (5.1) & \multirow{5}{*}{17.31} & \multirow{5}{*}{19.25} \\ \cline{1-5}
100 & 17.09 & [16.87-17.32] & [17.3-17.45] & 0.58 (3.4) &  &  \\ \cline{1-5}
150 & 17.25 & [17.08-17.42] & [17.41-17.52] & 0.44 (2.5) &  &  \\ \cline{1-5}
200 & 17.3 & [17.13-17.48] & [17.4-17.5] & 0.37 (2.1) &  &  \\ 
\bottomrule
\end{tabularx}
\raggedright \footnotesize{\noindent Note: All Obj, LB CI, UB CI and BG values are unitless.}
\end{table}

\begin{table}[t]
\setlength{\tabcolsep}{3.5pt}
\small
\centering
\caption{IEEE 123 system with 1200 scenarios for $\textbf{D}_\textbf{2}^\textbf{p}$}
\label{tab:IEEE 123 with 1200 scenario for power losss}
\begin{tabularx}{\columnwidth}{ccccccc}
\toprule
\multicolumn{1}{c}{$N_s$} & \multicolumn{1}{c}{Obj} & \multicolumn{1}{c}{LB CI} & \multicolumn{1}{c}{UB CI}  & \multicolumn{1}{c}{BG (\%)} & \multicolumn{1}{c}{EF} & \multicolumn{1}{c}{PH} \\ \midrule
50 & 48.53 & [47.04-50.02] & [50.26-50.62] & 3.58 (7.4) & \multirow{5}{*}{49.39} & \multirow{5}{*}{49.46} \\ \cline{1-5}
100 & 49.02 & [47.95-50.07] & [50.21-50.41] & 2.46 (5) &  &  \\ \cline{1-5}
150 & 49.7 & [48.9-50.52] & [50.33-50.83] & 1.93 (3.8) &  &  \\ \cline{1-5}
200 & 49.62 & [48.78-50.48] & [50.21-50.54] & 1.76 (3.5) &  &  \\ 
\bottomrule
\end{tabularx}
\raggedright \footnotesize{\noindent Note: All Obj, LB CI, UB CI and BG values are in kW.}
\end{table}

\subsection{Scalability assessment with system size}
We test the framework's scalability on the 9500-node system, a modified IEEE 8500-node test system, with a total load of 13,668 kW + 3,765 kvar~\cite{9500_nodes}, using a system base of 1 MVA and a voltage base of 12.47 kV. 
1200 scenarios are considered. 
EF would result in 112 million rows and 49.4 million columns, which is significantly challenging for any commercial solver to solve this optimization problem as a single large problem in polynomial time. Thus, EF and PH could not solve this problem as it became intractable using the considered workstation. Table~\ref{tab:9500 nodes with 1200 scenario for voltage deviation minimization}  and Table~\ref{tab:9500 nodes with 1200 scenario for power losss} present the results for the problem $\textbf{D}_\textbf{2}^\textbf{v}$ and $\textbf{D}_\textbf{2}^\textbf{p}$ respectively from SPAR-OPF. The obtained solutions are well within the tight bounds where the bounds gap is less than 3\% for the batch size $N_s \geq 150$. SPAR-OPF took 5.25 minutes and 5.5 minutes per seed to solve $\textbf{D}_\textbf{2}^\textbf{v}$ and $\textbf{D}_\textbf{2}^\textbf{p}$ respectively. This illustrates the scalability of proposed framework, making it suitable even for real-time applications under limited computing resources. 



\begin{table}[tbh]
\small
\setlength{\tabcolsep}{8pt}
\centering
\caption{9500-node system with 1200 scenarios for $\textbf{D}_\textbf{2}^\textbf{v}$}
\label{tab:9500 nodes with 1200 scenario for voltage deviation minimization}
\begin{tabularx}{\linewidth}{ccccc}
\toprule
\multicolumn{1}{c}{$N_s$} & \multicolumn{1}{c}{Obj} & \multicolumn{1}{c}{LB CI} & \multicolumn{1}{c}{UB CI} & \multicolumn{1}{c}{BG (\%)} \\ \midrule
50 & 622.7 & [609.9-635.5] & [622.5-639.1] & 29.2(4.68) \\
100 & 623.9 & [612.3-635.6] & [624.1-639.9] & 27.6 (4.42) \\
150 & 627.8 & [619.2-636.5] & [623.2-637] & 17.8 (2.8) \\
200 & 627.9 & [621.2-634.2] & [624.1-638] & 16.8 (2.7) \\
\bottomrule
\end{tabularx}
\raggedright \footnotesize{\noindent Note: All Obj, LB CI, UB CI and BG values are unitless.}
\end{table}

\begin{table}[tbhp]
\setlength{\tabcolsep}{8pt}
\small
\centering
\caption{9500-node system with 1200 scenarios for $\textbf{D}_\textbf{2}^\textbf{p}$}
\label{tab:9500 nodes with 1200 scenario for power losss}
\begin{tabularx}{\linewidth}{ccccc}
\toprule
\multicolumn{1}{c}{$N_s$} & \multicolumn{1}{c}{Obj} & \multicolumn{1}{c}{LB CI} & \multicolumn{1}{c}{UB CI} & \multicolumn{1}{c}{BG (\%)} \\ \midrule
50 & 357 & [349.2-364.8] & [367.2-370.9] & 21.7 (6.1) \\
100 & 361.2 & [355.5-366.9] & [368.2-371.9] & 16.4 (4.5)\\
150 & 364.1 & [360.2-367.9] & [367.7-371.1] & 10.9 (3.0) \\
200 & 364.8 & [361.2-368.4] & [367.9-371.8] & 10.6 (2.9) \\
\bottomrule
\end{tabularx}
\raggedright \footnotesize{\noindent Note: All Obj, LB CI, UB CI and BG values are in kW.}
\end{table}

\subsection{Computational comparison between  $\textbf{L}_\textbf{1}$ and  $\textbf{E}_\textbf{1}$}
This section compares 
the computational performance of two formulations, $\textbf{L}_\textbf{1}$ and  $\textbf{E}_\textbf{1}$. Table~\ref{tab:time comparison between lambda and epigraph method} reports results for different test cases and objectives, where columns $\textbf{D}_\textbf{2}^\textbf{v}$ and $\textbf{D}_\textbf{2}^\textbf{p}$ denote per-iteration solve times.
The number of rows and columns for $\textbf{D}_\textbf{2}^\textbf{v}$ is slightly higher due to additional constraint~(\ref{eq: objective reformulation constraint}) and number for $\textbf{D}_\textbf{2}^\textbf{v}$ is reported. Leveraging the convexity of value function and using only a few additional continuous variables,  $\textbf{E}_\textbf{1}$ is observed to perform better (4x less time) than  $\textbf{L}_\textbf{1}$. Even for larger system size,  $\textbf{E}_\textbf{1}$ required only 5 seconds per iteration. Based on all simulations, 80 iterations were required on average for convergence. This confirms that the SPAR-OPF applies to realistic large-scale power systems.  


\begin{table}[tbhp]
\centering
\small
\setlength{\tabcolsep}{2pt}
\caption{Computational comparison between $\textbf{L}_\textbf{1}$ and $\textbf{E}_\textbf{1}$}
\label{tab:time comparison between lambda and epigraph method}
\begin{tabularx}{\columnwidth}{cccccc} \toprule
System & Method & $\textbf{D}_\textbf{2}^\textbf{v}(s)$ & $\textbf{D}_\textbf{2}^\textbf{p}(s)$ & Rows & Columns \\ \midrule
\multirow{2}{*}{\begin{tabular}[c]{@{}c@{}}IEEE 123 bus \\ system\end{tabular}} &  $\textbf{L}_\textbf{1}$ & 2.08 & 2.1 & 524 &  21.9k (21.8k binary)\\
 &  $\textbf{E}_\textbf{1}$ & 0.67 & 0.69 & 21.9k & 391 (130 binary) \\ \midrule
\multirow{2}{*}{\begin{tabular}[c]{@{}c@{}}9500-node \\ system\end{tabular}} &  $\textbf{L}_\textbf{1}$ & 20.1 & 20.3 & 10.9k & 463.7k (414.3k binary) \\
 &  $\textbf{E}_\textbf{1}$ & 5 & 5.03 & 463.7k & 8.2k (2.7k binary)\\ \bottomrule
\end{tabularx}
\end{table}

\section{Conclusion and Future Work} \label{sec:Conclusion}

We developed the SPAR-OPF framework to solve OPF-based two-stage stochastic optimization problems in power systems. To the best of our knowledge, this is the first application of separable projective approximations of the value function to such problems. 
We provided two formulations for learned functions and a statistical validation procedure. The framework is demonstrated on two standard three-phase unbalanced distribution system test cases for DG siting and sizing problem, showing near-optimal solutions with significantly lower solve times compared to state-of-the-art methods.
The scalability of the framework is validated on a larger system with a higher number of scenarios, making it suitable even for resource-constrained environments.
As future work, we will explore heuristics such as  MIP relaxation, exploration vs. exploitation, and successive MIP gap reduction to further reduce solution time and extend the framework to classical stochastic problems such as unit commitment and dynamic economic dispatch.

\ifCLASSOPTIONcaptionsoff
  \newpage
\fi

\bibliographystyle{IEEEtran}
\end{document}